\documentclass[aip,jmp]{revtex4-1}
\usepackage[utf8]{inputenc}
\usepackage{amsfonts}
\usepackage{amsmath}
\usepackage{amssymb}
\usepackage{amsthm}
\usepackage{mathrsfs}
\usepackage{latexsym}
\usepackage{enumerate}
\usepackage{graphicx}

\date{April 12, 2018}

\newtheorem{theorem}{Theorem}[section]
\newtheorem*{theorem-non}{Theorem}

\newtheorem{lemma}[theorem]{Lemma}
\newtheorem{corollary}[theorem]{Corollary}
\newtheorem{proposition}[theorem]{Proposition}

\theoremstyle{definition}
\newtheorem{remark}[theorem]{Remark}

\newcommand{\ee}{\mathrm{e}}
\newcommand{\dd}{\mathrm{d}}
\newcommand{\R}{{\mathord{\mathbb R}}}
\newcommand{\Q}{{\mathord{\mathbb Q}}}

\newcommand{\N}{{\mathord{\mathbb N}}}

\newcommand{\dom}[1]{\mathrm{Dom}(#1)}

\newcommand{\vertiii}[1]{{\left\vert\kern-0.25ex\left\vert\kern-0.25ex\left\vert #1
    \right\vert\kern-0.25ex\right\vert\kern-0.25ex\right\vert}}

\DeclareMathOperator*{\esssup}{ess\,sup}
\DeclareMathOperator*{\essinf}{ess\,inf}

\begin{document}

\title[A geometric Iwatsuka type effect in quantum layers]{A geometric Iwatsuka type effect in quantum layers}

\author{Pavel Exner}
\affiliation{Doppler Institute for Mathematical Physics and Applied Mathematics, Prague, Czechia}
\affiliation{Department of Physics, Faculty of Nuclear Sciences and Physical Engineering, Czech Technical University in Prague, Prague, Czechia}
\affiliation{Nuclear Physics Institute, Czech Academy of Sciences, \v{R}e\v{z}, Czechia}
\email{exner@ujf.cas.cz}

\author{Tom\'a\v{s} Kalvoda}
\affiliation{Department of Applied Mathematics, Faculty of Information Technology, Czech Technical University in Prague, Prague, Czechia}
\email{tomas.kalvoda@fit.cvut.cz}

\author{Mat\v{e}j Tu\v{s}ek}
\affiliation{Department of Mathematics, Faculty of Nuclear Sciences and Physical Engineering, Czech Technical University in Prague, Prague, Czechia}
\email{tusekmat@fjfi.cvut.cz}

\begin{abstract}
\noindent We study motion of a charged particle confined to
Dirichlet layer of a fixed width placed into a homogeneous magnetic
field. If the layer is planar and the field is perpendicular to it
the spectrum consists of infinitely degenerate eigenvalues. We
consider translationally invariant geometric perturbations and
derive several sufficient conditions under which a magnetic
transport is possible, that is, the spectrum, in its entirety or a
part of it, becomes absolutely continuous.
\end{abstract}

\maketitle

\section{Introduction} \label{s:intro}

A homogeneous magnetic field acting on charged particles has a
localizing effect, both classically and quantum mechanically. Since
numerous physical effects are based on moving electrons between
different places, mechanisms that can produce transport in the
presence of a magnetic field are of great interest. They typically
require presence of an infinitely extended perturbation, a standard
example being a barrier or a potential wall producing edge states,
cf. Ref.~\onlinecite{DoGeRa_11, DoHiSo_14, GeSe_97, HiPoRaSu_16, HiSo_08a, HiSo_08, HiSo_15} and references therein. This
is not the only possibility, though. In his classical paper
\cite{Iw_85} Iwatsuka demonstrated that a transport can be induced
by a modification of the magnetic field itself under the assumption
of a translational invariance, see also Ref.~\onlinecite{CFKS}. This is what we call
the Iwatsuka effect.

The aim of the present paper is to show still another mechanism which
can produce transport in a homogeneous field for particles confined
to a layer with hard walls. As in the case of the Iwatsuka model we
will express the effect in spectral terms seeking perturbations that
change the Hamiltonian spectrum to absolutely continuous. Our departing point is a flat layer of width $2a$ to which a charged particle is confined and which is exposed to the homogeneous magnetic field perpendicular to the layer plane. The spectrum of such a system is easily found by separation of variables. It combines the Landau levels with the Dirichlet Laplacian eigenvalues in the perpendicular direction, and needless to say that all the resulting eigenvalues are infinitely degenerate, see Sec.~\ref{sss: flatlayer} below for more details.

We are going to discuss \emph{geometric perturbations} of such a system, in particular, deformations of the layer which are invariant with respect to translation in a fixed direction. Such layers can be described, e.g., as a set of points satisfying $\mathrm{dist}(x,\Sigma)<a$ where $\Sigma$ is a surface obtained by shifting a smooth curve which can be parametrized by relation \eqref{surface} below. We are going to derive several conditions which ensure that the unperturbed pure point spectrum will change into an absolutely continuous one. Let us stress that this effect is purely geometrical and has not been described so far. As a tribute to the original work of Iwatsuka, we decided to call it the geometric Iwatsuka type effect. Our main results can be summarized in the following assertion.

\begin{theorem}
Let $H$ be the Hamiltonian of a charged quantum particle confined to a layer $\Omega$ of a constant width $2a$ in $\R^3$ built over a $C^4$-smooth, translationally invariant surface \eqref{surface}, $(s,y)\mapsto(x(s),y,z(s))$, and exposed to a nonzero homogeneous magnetic field pointing in the $z$-direction. The spectrum of $H$ is purely absolutely continuous if together with technical assumptions \eqref{eq:curv_bound} and \eqref{eq:diff_cond} any of the following conditions is satisfied:
\vspace{-.3em}
\begin{enumerate}[(i)]
\setlength{\itemsep}{-3pt}
 \item $\Omega$ is a one-sided-fold layer, $\lim_{s\to\pm\infty}x(s)=+\infty$ or $\lim_{s\to\pm\infty}x(s)=-\infty$. Furthermore, we suppose that the second part of \eqref{eq:V_bound} is fulfilled. (See Fig.~\ref{fig:one_sided_fold} for an example.)
 \item $\Omega$ is bent and asymptotically flat, $\dot{x}(s)=\alpha_{\pm}$
for all large enough positive and negative $s$, respectively, where
$\alpha_\pm\in(0,1]$. Furthermore, one requires that $\alpha_{+}\neq\alpha_{-}$ and
the halfwidth $a$ satisfies the bound described in
Lemma~\ref{lem:ev_diff} and Remark~\ref{rem:a_crit_bound} below.
(See Fig.~\ref{fig:asy_flat} for examples.)
\end{enumerate}

\vspace{-.6em}

\noindent Moreover, for a fixed $E\in\R$, the spectrum of $H$ below
$E+\left(\frac{\pi}{2a}\right)^2$ is absolutely continuous if
$\Omega$ is thin, i.e. the halfwidth $a$ is sufficiently small, and
the generating surface satisfies additional conditions specified in
Proposition~\ref{prop:thin}.
\end{theorem}

\begin{figure}[htb]
\centering
  \includegraphics[width=6cm]{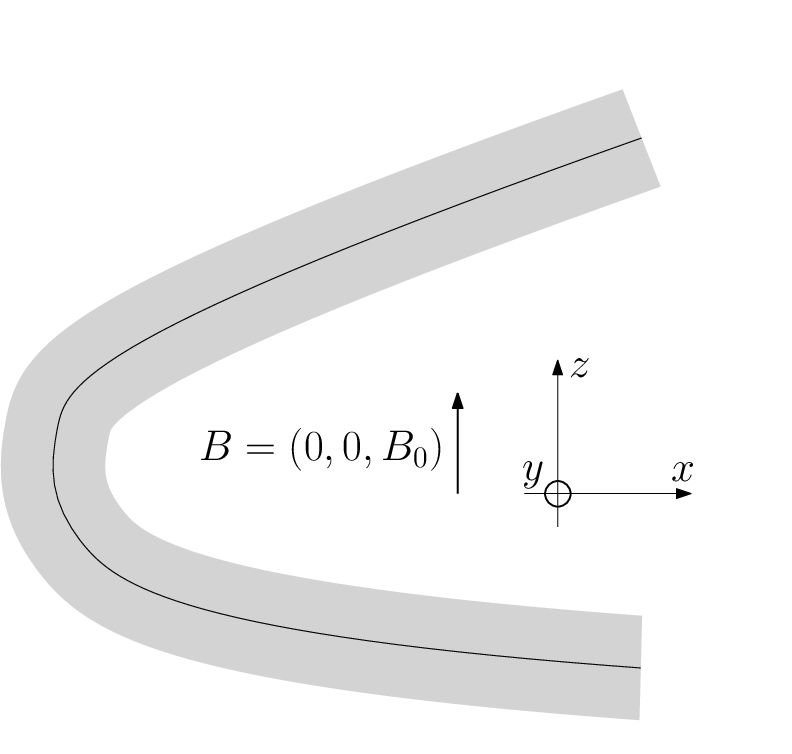}
  \caption{Example of a profile of a one-sided-fold layer.}
  \label{fig:one_sided_fold}
\end{figure}  

\begin{figure}[ht]
\centering
  \includegraphics[width=9cm]{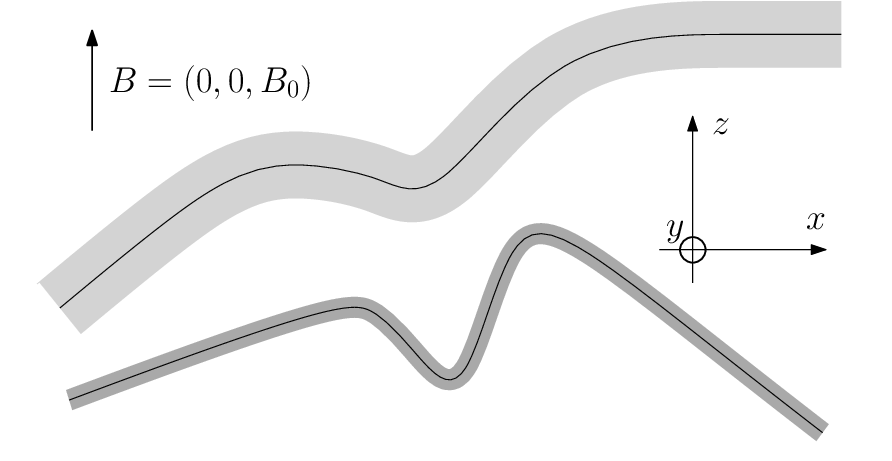}
  \caption{Examples of profiles of asymptotically flat layers. The larger is difference of \emph{magnitudes} of asymptotic slopes at $x=\pm\infty$ the thicker layer may be considered.}
  \label{fig:asy_flat}
\end{figure}
 
The proof of the theorem will be given in Sec.~\ref{s:classes}, before coming to it we will describe the geometry of the layer and explain the main steps of the argument. Let us add a few remarks. First of all, in Sec.~\ref{sss: flatlayer} we demonstrate that the perturbation must be geometrically nontrivial, because a mere tilt of the layer with respect to the field direction is not enough, with one notable exception. Furthermore, except the claim (i) our condition impose restrictions on the layer thickness. On the other hand, the thinner the layer, the more general deformations we can treat. In particular, the last claim covers perturbations which are compact with respect to the $x$ variable, cf. Proposition \ref{prop:thin}. Note also that the shift by $\left(\frac{\pi}{2a}\right)^2$ in the last claim is needed; without it the claim would be trivial because in a thin layer the spectral threshold is pushed up due to the Dirichlet boundaries.

The method used to treat thin layers is also useful with respect to the original Iwatsuka model and its generalization including addition of a potential perturbation. Recall, in particular, the conjecture stated in Ref.~\onlinecite{CFKS} according to which \emph{any} nontrivial translationally invariant magnetic perturbation gives rise to the purely absolutely continuous spectrum. Despite a number of sufficient conditions derived after the original Iwatsuka paper \cite{MaPu_97, ExKo_00, Tu_16} to which we add a new one in Theorem~\ref{theo:iwatsuka}, the question in its generality remains open. In a similar vein, we are convinced that the sufficient conditions we find in this paper are by far not necessary. Particularly, any new result on the original two-dimensional Iwatsuka Hamiltonian with an additional translationally invariant scalar potential may give rise to a new sufficient condition on absolute continuity of the Hamiltonian on a thin layer, cf. Remark \ref{rem:eff}.

\section{Preliminaries} \label{s:prelim}

\subsection{Geometry of the layer}

Let $\Sigma$ be a surface in $\R^3$ invariant with respect to
translation in the $y$ direction and described by means of the
following parametrization,
\begin{equation} \label{surface}
\mathscr{L}_{0}(s,y)=(x(s),y,z(s))
\end{equation}
with $s,y\in\R$. The functions $x$ and $z$ here are assumed to be smooth enough, unless said otherwise we suppose they are $C^{4}$, and  such that 
\begin{equation} \label{eq:param}
\dot x(s)^2+\dot z(s)^2=1,
\end{equation}
where the dot stands for the derivative with respect to $s$. The last condition means that the curve $\Gamma:s\mapsto(x(s),z(s))$ in the $xz$ plane is parametrized by its arc length measured from some reference point on the curve. Therefore the (signed) curvature $\kappa$ of $\Gamma$ is given by
$$
\kappa(s)=\dot x(s)\ddot z(s)-\ddot x(s)\dot z(s)
$$
and the corresponding unit normal vector to $\Sigma$ is
$$
n(s,y)\equiv n(s)=(-\dot z(s),0,\dot x(s)).
$$
Let us stress that throughout the paper we assume that
\begin{equation}\label{eq:curv_bound}
 \|\kappa\|_\infty<\infty.
\end{equation}
If we regard $\Sigma$ as a Riemannian manifold, then the metric induced by the immersion $\mathscr{L}_{0}$ is
\begin{equation*}
 (g_{\mu\nu})=\begin{pmatrix}
               1 & 0\\
               0 & 1
              \end{pmatrix}; \quad \mu,\nu\in\{s,y\}.
\end{equation*}
Let $a>0$ and $I:=(-1,1)$. We define the layer $\Omega$ of width $2a$ built over the surface $\Sigma$ as the image of
$$
\mathscr{L}:\R^2\times I\to\R^3:\ \left\{ (s,y,u)\mapsto \mathscr{L}_{0}(s,y)+a u n(s)\right\},
$$
see Fig.~\ref{fig:layer}.
\begin{figure}[h] 
\centering
\includegraphics[scale=0.7, trim=0 1cm 0 0]{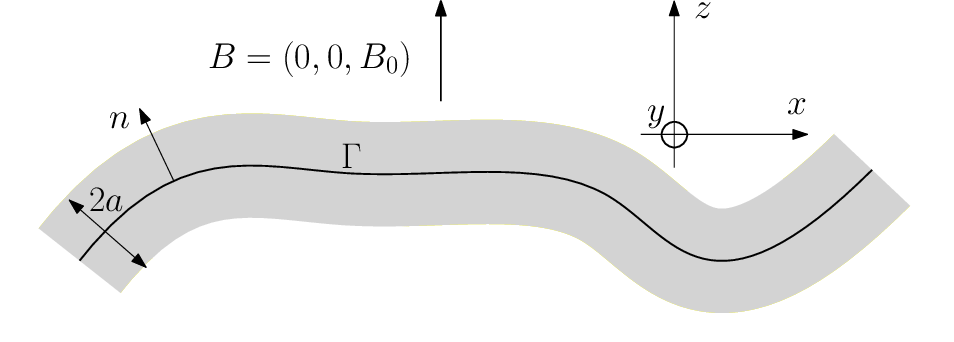} 
\caption{A piece of profile of a layer $\Omega$. The both ends extend infinitely.} \label{fig:layer}
\end{figure}
We will always assume that
\begin{equation}\label{eq:diff_cond}
 a<\varrho_{m}:=\|\kappa\|_{\infty}^{-1}\quad\text{and }\Omega\text{ does not intersect itself}.
\end{equation}
Under these conditions, $\mathscr{L}$ is a diffeomorphism onto $\Omega$ as one can see, e.g., from the formula for the metric $G$ on $\Omega$, induced by the immersion $\mathscr{L}$, that reads
\begin{equation*}
 (G_{ij})=\begin{pmatrix}
           (G_{\mu\nu})& 0\\
           0 & a^2
          \end{pmatrix},\quad
 (G_{\mu\nu})=\begin{pmatrix}
               f_a(s,u)^2 & 0\\
               0 & 1
              \end{pmatrix};\quad i,j\in\{s,y,u\},
\end{equation*}
where $f_a(s,u):=1-a u\kappa(s)$. The assumption \eqref{eq:diff_cond}
implies, in particular, that $f_a(s,u) > 1 - a\|\kappa\|_\infty > 0$
holds for all $(s,u) \in \mathbb{R}\times I$.
\begin{remark}
Note that one can use $v=au \in (-a,a)$ as a natural transverse variable. The choice we made is suitable in situations when we want to discuss asymptotic properties of thin layers.
\end{remark}

\subsection{Dirichlet magnetic Laplacian}

The main object of our interest is the magnetic Laplacian on $\Omega$ subject to the Dirichlet boundary condition,
$$
-\Delta_{D,A}^{\Omega}=(-i\nabla+A)^2\text{ (in the form sense)},\quad Q(-\Delta_{D,A}^{\Omega})=\mathcal{H}_{A,0}^{1}(\Omega,\dd x\dd y\dd z),
$$
with a special choice of the vector potential, $A=(0,B_{0} x,0),\ B_{0}>0$, that corresponds to the homogeneous magnetic field $B=(0,0,B_{0})$. Using the unitary transform $\tilde U:\, L^2(\Omega,\dd x\dd y\dd z)\to L^2(\R^2\times I,\dd \Omega)$, $\,\psi\mapsto\psi\circ\mathscr{L}$, we may identify $-\Delta_{D,A}^{\Omega}$ with the self-adjoint operator $\hat{H}$ defined, in the form sense, on $L^{2}(\R^2\times I,\dd\Omega)$ by
\begin{align*}
\hat{H} = &-f_a(s,u)^{-1}\partial_{s}f_a(s,u)^{-1}\partial_{s}+(-i\partial_{y}+\tilde{A}_{2}(s,u))^2 \\
&-a^{-2}f_a(s,u)^{-1}\partial_{u}f_a(s,u)\partial_{u},
\end{align*}
where $\tilde{A}=(D\mathscr{L})^T A\circ\mathscr{L}=(0,\tilde{A}_{2},0)$ with
$$
\tilde{A}_{2}(s,u)=B_{0}\big(x(s)-a u\dot z(s)\big).
$$
By another unitary transform, $U:L^{2}(\R^2\times I,\dd\Omega)\to L^{2}(\R^2\times I,\dd\Sigma\dd u)$, $\,\psi\mapsto a^{1/2}f_{a}^{1/2}\psi$, we pass to the unitarily equivalent operator defined, again in the form sense, as
\begin{equation*}
 \tilde{H}=U\hat{H}U^{-1}=-\partial_{s}f_a(s,u)^{-2}\partial_{s}+(-i\partial_{y}+\tilde{A}_{2}(s,u))^2-a^{-2}\partial_{u}^{2}+V(s,u),
\end{equation*}
where
$$
V(s,u)=-\frac{1}{4}\frac{\kappa(s)^2}{f_a(s,u)^2}-\frac{1}{2}\frac{a u\ddot\kappa(s)}{f_a(s,u)^3}-\frac{5}{4}\frac{a^2 u^2\dot\kappa(s)^2}{f_a(s,u)^4}.
$$
These formul{\ae} are easy to derive, see e.g. Ref.~\onlinecite{DuExKr_01,KrRaTu_15} or Ref.~\onlinecite[Sec.~1.1]{EK}. Remark that we needed $C^4$--smoothness of $\Sigma$ to write $V$ down as an operator. Nevertheless,  $\tilde H$ may be introduced via its quadratic form for any $C^3$ surface.

The translational invariance of $\Omega$ makes it possible to pass finally to still another unitarily equivalent form of the operator by means of the
Fourier--Plancherel transform in the $y$ variable,
$$
(\mathscr{F}_{y\to\xi}\psi)(s,\xi,u)=(2\pi)^{-1/2}\int_{\R}\mathrm{e}^{-iy\xi}\,\psi(s,y,u)\,\dd y,
$$
which yields
\begin{equation*}
H:=\mathscr{F}_{y\to\xi}\tilde{H}\mathscr{F}_{y\to\xi}^{-1}=-\partial_{s}f_a(s,u)^{-2}\partial_{s}+(\xi+\tilde{A}_{2}(s,u))^2-a^{-2}\partial_{u}^{2}+V(s,u).
\end{equation*}
For a fixed $\xi\in\R$, we define
\begin{equation}\label{eq:part_hamiltonian}
 H[\xi]:=-\partial_{s}f_a(s,u)^{-2}\partial_{s}+(\xi+\tilde{A}_{2}(s,u))^2-a^{-2}\partial_{u}^{2}+V(s,u),
\end{equation}
which allows us to write our Hamiltonian in the form of a direct integral,
\begin{equation} \label{eq:di_decomp}
 H=\int^\oplus_\R H[\xi]\,\dd\xi,
\end{equation}
where $\xi$ is the momentum of the motion in the $y$ direction. Note that since $\dd\Sigma=\dd s\wedge\dd y$, the operator $H$ and its fiber $H[\xi]$ act in $L^{2}(\R^2\times I,\dd s\dd\xi\dd u)$ and $L^{2}(\R\times I,\dd s\dd u)$, respectively.

\section{Absolute continuity of the magnetic Laplacian} \label{s:ac}

As mentioned in the introduction we are interested in situations when the confinement causes a magnetic transport manifested through the absolute continuity of the spectrum. Our aim is to describe several classes of layers $\Omega$ for which the spectrum $-\Delta_{D,A}^{\Omega}$ is purely absolutely continuous. Since this operator is unitarily equivalent to the above described $H$ which in turn decomposes into a direct integral with fibers $H[\xi]$, by Ref.~\onlinecite[Thm XIII.86]{RS4}, it is sufficient to prove that
\begin{enumerate}[(a)]
 \item the family $\{H[\xi] \mid \xi\in\R\}$ is analytic with respect to $\xi$ in the sense of Kato, \label{suff_1}
 \item the resolvent of $H[\xi]$ is compact for any $\xi\in\R$, \label{suff_2}
 \item no eigenvalue branch of $H[\xi]$ is constant in the variable $\xi$. \label{suff_3}
\end{enumerate}
In the following subsections, we address consecutively each of these points.
\begin{remark}\label{rem:spec_prop}
In fact, to see that \eqref{suff_1}--\eqref{suff_3} are sufficient conditions for the absolute continuity of $H$, we have to modify Theorem XIII.86 in Ref.~\onlinecite{RS4} slightly. In particular, the variable $\xi$ in the direct integral \eqref{eq:di_decomp} runs over a non-compact interval and the eigenvalue branches are not necessarily bounded. An alternative short proof is based on the main result of Ref.~\onlinecite{FiSo_06} which says that under \eqref{suff_1} and \eqref{suff_2}, $\sigma_{sc}(H)=\emptyset$ and $\sigma_{p}(H)$ may consist of isolated points without finite accumulation point only, and moreover, each one of these eigenvalues is of infinite multiplicity.

Since there are no finite accumulation points in $\sigma_{p}(H)$ and $\sigma_{sc}(H)=\emptyset$, we have $\sigma(H)=\sigma_{p}(H)\cup\sigma_{ac}(H)$, and consequently, it remains to check that $\sigma_{p}(H)=\emptyset$. By Ref.~\onlinecite[Thm XIII.85]{RS4}, $\lambda$ is an eigenvalue of $H$ if and only if $\mathrm{meas}_1\{\xi\in\R|\, \lambda\in\sigma_{p}(H[\xi])\}>0$, where $\mathrm{meas}_1$ stands for the one-dimensional Lebesgue measure. Since $H[\xi]$ has compact resolvent by \eqref{suff_2} and it is unbounded but lower bounded, its spectrum consists of a sequence of eigenvalues of finite multiplicity which can accumulate only at $+\infty$. This together with \eqref{suff_1} implies that there are countably many eigenvalue branches. Consequently, $\lambda\in\sigma_{p}(H)$ if and only if there is an eigenvalue branch, say $\lambda_p[\xi]$, such that $\mathrm{meas}_1\{\xi\in\R|\, \lambda_p[\xi]=\lambda\}>0$. By \eqref{suff_1}, the function $\lambda_p[\cdot]-\lambda$ is real analytic, and by \eqref{suff_3}, 
it is non-constant.
 This means that the equation $\lambda_p[\xi]=\lambda$ may have at most countable number of solutions and proves thus the claim.
\end{remark}

\subsection{Analyticity in $\xi$} \label{sec:analyticity}

For any $\xi_0\in\R$, we have in the form sense
$$
H[\xi]=H[\xi_0]\dotplus p_{\xi},$$
where the quadratic form $p_{\xi}$ is given by
\begin{equation}\label{eq:xi_pert}
p_{\xi}(\psi)=(\xi-\xi_0)^2 \|\psi\|^2+2(\xi-\xi_0)\langle\psi,(\xi_0+\tilde{A}_{2})\psi\rangle.
\end{equation}
For any $\delta>0$, one easily gets from here
\begin{equation*}
\begin{split}
 |p_{\xi}(\psi)|&\leq (\xi-\xi_0)^2 (1 +\delta^{-1})\|\psi\|^2+\delta\|(\xi_0+\tilde{A}_{2})\psi\|^2\\
 &\leq (\xi-\xi_0)^2 (1+\delta^{-1})\|\psi\|^2 +\delta\langle\psi,H[\xi_0]\psi\rangle+\delta\langle\psi,V_{-}\psi\rangle,
\end{split}
\end{equation*}
where $V_{-}$ stands for the negative part of $V =: V_{+} - V_{-}$. If we assume that
\begin{equation}\label{eq:neg_part_assump}
 V_{-}\text{ is relatively form bounded by } H[0],
\end{equation}
which is equivalent to the assumption that $V_{-}$ is relatively form bounded by $H[\xi_0]$, then  $H[\xi]$ is lower bounded and $p_{\xi}$ is infinitesimally form bounded by $H[\xi_0]$. In combination with (\ref{eq:xi_pert}) this implies that $H[\xi]$ forms an analytic family of type (B), in particular, that $H[\xi]$ is an analytic family in the sense of Kato~\cite{kato}.

\subsection{Compactness of the resolvent}\label{sec:comp}

Assume now, in addition, that
\begin{equation}\label{eq:V_bound}
\big|\!\lim_{s\to\pm\infty}x(s)\big|=+\infty,\quad V_{-}\in L^{\infty}(\R\times I,\dd s\dd u).
\end{equation}
Remark that under the second assumption, \eqref{eq:neg_part_assump} is trivially satisfied. 
By Ref.~\onlinecite[Thm XIII.64]{RS4}, the fiber $H[\xi]$ has compact resolvent if and only if the set
$$
\mathscr{C}_{H[\xi],b}:=\{ \psi\in Q(H[\xi]) \mid \|\psi\|\leq 1,\, \langle\psi, H[\xi]\psi\rangle\leq b\}
$$
is compact for all $b$. By \eqref{eq:diff_cond} we have $f_a(s,u)\leq 1+a\|\kappa\|_\infty=:d$, and therefore
$$
H[\xi]\geq -d^{-2}\partial_{s}^{2}+(\xi+\tilde{A}_2)^2-a^{-2}\partial_{u}^{2}-\|V_{-}\|_\infty.
$$
Moreover, there exists clearly an $\tilde s$ such that for all $s:\,|s|\geq \tilde s$, we have
$$
(\xi+\tilde{A}_2)^2\geq \left(\xi+\frac{B_0}{2}x(s)\right)^2.
$$
If we introduce the constant
$$
K:=\sup_{|s|<\tilde s,\, u\in I}\left|(\xi+\tilde{A}_2(s,u))^2-\left(\xi+\frac{B_0}{2}x(s)\right)^2\right|,
$$
then we can estimate $H[\xi]$ from below as follows,
$$
H[\xi]\geq -d^{-2}\partial_{s}^{2}+\left(\xi+\frac{B_0}{2}x(s)\right)^2-a^{-2}\partial_{u}^{2}-K-\|V_{-}\|_\infty=:H_{-}-K-\|V_{-}\|_\infty.
$$
The operator $H_{-}$ on the right-hand side decomposes into a sum of a one-dimensional operator which has compact resolvent due to \eqref{eq:V_bound}, cf. Ref.~\onlinecite[Thm XIII.67]{RS4}, and the Dirichlet Laplacian on $I$. Hence, according to Ref.~\onlinecite[Thm XIII.64]{RS4}, it has compact resolvent too. Consequently, $\mathscr{C}_{H_{-},b}$ is compact for any $b$. Clearly, $\mathscr{C}_{H[\xi],b}\subset\mathscr{C}_{H_{-},b+K+\|V_{-}\|_\infty}$, which means that $\mathscr{C}_{H[\xi],b}$ has to be precompact, but at the same time $\mathscr{C}_{H[\xi],b}$ is closed, cf. the proof of Theorem XIII.64 in Ref.~\onlinecite{RS4}, and thus compact.

\subsection{Non-constancy of the eigenvalues}

In the following section, we will study this property for several special but still rather wide classes of layers indicated in the introduction. Specifically, we will be concerned with the following cases:
\begin{enumerate}[(i)]
 \item \emph{a one-sided-fold layer:} $\lim_{s\to\pm\infty}x(s)=+\infty$ or $\lim_{s\to\pm\infty}x(s)=-\infty$,

 \item \emph{a bent, asymptotically flat layer:} $\dot{x}(s)=\alpha_+$ for all large enough positive $s$ and $\dot{x}(s)=\alpha_-$ for all large enough negative $s$, where $\alpha_\pm\in(0,1],\, \alpha_+\neq\alpha_-$,
 \item \emph{a thin non-planar layer:} $a$ is sufficiently small.
\end{enumerate}
Let us remark that there are basically two methods how to demonstrate non-constancy of the eigenvalues. The first relies on the Feynman-Hellmann formula that gives the derivative of an eigenvalue with respect to a parameter. It is useful in situations when the curvature is compactly supported. The other method is based on some type of a comparison argument that should give asymptotic behavior of the eigenvalues at $\pm\infty$. It is usually applied in situations when the curvature behaves differently at $\pm\infty$.

\section{Special classes of layers} \label{s:classes}
Recall that throughout the paper we assume \eqref{eq:curv_bound} and \eqref{eq:diff_cond}. In this section, we will suppose that \eqref{eq:V_bound} (which in turn implies \eqref{eq:neg_part_assump}) is always satisfied, too. 

\subsection{One-sided-fold layer}\label{sec:horseshoe}

For the sake of definiteness, assume that  $\lim_{s\to\pm\infty}x(s)=+\infty$. Then
$$
\lim_{s\to\pm\infty}\tilde{A}_2 (s,u)=+\infty
$$
holds for all $u\in I$. Using \eqref{eq:param} we obtain $\tilde{A}_2(s,u) \geq B_0(x(s) - a)$. This in turn implies that $(\tilde{A}_2)_{-}$ is compactly supported. For any $\xi>0$ we have
$$
(\xi+\tilde{A}_2)^2\geq\xi^2+\tilde{A}_{2}^{2}-2\xi\|(\tilde{A}_2)_{-}\|_\infty,
$$
and therefore
$$
H[\xi]\geq -\partial_s f_{a}^{-2}\partial_s-a^{-2}\partial_{u}^{2}+\tilde{A}_{2}^{2}+\xi^2-2\xi\|(\tilde{A}_2)_{-}\|_\infty-\|V_{-}\|_\infty.
$$
The first three terms on the right-hand side are positive and for the remaining part, independent of $s$ and $u$, we have
$$
\lim_{\xi\to+\infty}\left(\xi^2-2\xi\|(\tilde{A}_2)_{-}\|_\infty-\|V_{-}\|_\infty\right)=+\infty.
$$
Thus to any $C>0$ there is a $\xi_C\in\R$ such that $H[\xi]>C$ holds for all $\xi>\xi_C$, and consequently, no eigenvalue branch may be constant as a function of $\xi$. This together with the general results of Sec.~\ref{sec:analyticity} and \ref{sec:comp} means that the spectrum of $H$ is purely absolutely continuous. Due to unitary equivalence, the same holds true for $-\Delta_{D,A}^\Omega$.

\subsection{A digression: flat layers}

Before proceeding further we are going to show that a mere rotation
of the layer around an axis perpendicular to the magnetic field is
not sufficient -- with one notable exception -- to produce a
transport in the considered system. The decisive quantity is the
tilt angle between the field direction and the layer.

\subsubsection{Inclined layer not parallel with the magnetic field}
\label{sss: flatlayer}

Let $\gamma$ stand for the angle between the magnetic field $B$ and the normal vector to the layer $n$. In the case of the unperturbed system, a planar layer with a
perpendicular field ($\gamma=0$), it is straightforward
to see by separation of variables in the cylindrical coordinates that the spectrum of $H$ is purely
point. The same is true if $\dot{x}(s)=\cos{\gamma},\
\gamma\in(-\frac{\pi}{2},\frac{\pi}{2})\setminus\{0\},$ holds for all $s\in\R$. Indeed, each
of the fiber operators
$$
H[\xi]=-\partial_{s}^{2}+\left(\xi+B_0(s \cos{\gamma}-au\sin{\gamma})\right)^2-a^{-2}\partial_{u}^{2}
$$
is unitarily equivalent to $H[0]$, as one can verify employing a
unitary transform $\psi(s,u)\mapsto \psi(s+\xi/(B_0
\cos{\gamma}),u).$ Using then Ref.~\onlinecite[Thm XIII.85]{RS4}, in combination
with the fact that $H[0]$ has compact resolvent, we conclude that
$\sigma(H)=\sigma_p(H)=\sigma(H[0])$ consists of infinitely
degenerate eigenvalues.

For $\gamma=0$, this yields an alternative way to determine the
spectral character in the unperturbed case by observing that
\begin{equation*}
 H[\xi]=-\partial_{s}^{2}+(\xi+B_{0}s)^2-a^{-2}\partial_{u}^{2}
\end{equation*}
decomposes then into the sum of the Hamiltonian of the harmonic
oscillator with the origin shifted by $-\xi/B_{0}$ and the Dirichlet
Hamiltonian on the line segment $I$. For its eigenpairs,
$(\lambda_{m,n}[\xi],\psi_{m,n}[\xi])$, we have
\begin{align*}
 &\lambda_{m,n}[\xi]\equiv \lambda_{m,n}=B_{0}(2m+1)+\left(\frac{n\pi}{2a}\right)^2\\
 &\psi_{m,n}[\xi](s,u)=\psi_{m}(s+\xi/B_{0})\chi_{n}(u),
\end{align*}
with
\begin{equation}\label{eq:HOef}
 \psi_{m}(x)=(2^m m!)^{-1/2}(B_{0}/\pi)^{1/4}\ee^{-B_{0}x^2/2}H_{m}(B_{0}^{1/2}x),
\end{equation}
where $H_{m}$ stands for the $m$th Hermite polynomial, and
\begin{equation*}
 \chi_{n}(u)=\begin{cases}
              \cos(n\pi u/2)&\text{if }n\text{ is odd}\\
        \sin(n\pi u/2)&\text{if }n\text{ is even.}
             \end{cases}
\end{equation*}
Here $m\in\N_{0},\, n\in\N$ and all the eigenfunctions are
normalized to one in the respective Hilbert spaces.

Since $\lambda_{m,n}[\xi]$ is independent of $\xi$, all the
eigenvalues of $H$ have infinite multiplicity. Let us ask about additional
degeneracies, that is, about the multiplicity of eigenvalues of
$H[\xi]$. Assume that $\lambda_{m,n} =\lambda_{\tilde{m},\tilde{n}}$
for some $m,\tilde{m}\in\N_{0}$ and $n,\tilde{n}\in\N$ such that
$m\neq\tilde{m}$ and $n\neq\tilde{n}$. This means that
$$
\theta m+n^2=\theta\tilde{m}+\tilde{n}^2,
$$
where $\theta:=8 B_{0}(a/\pi)^2$, which implies that $\theta$ is a positive rational, $\theta\in\Q_+$. Conversely, if this is the case then $\theta=p/q$, for some
$p,q\in\N$, and the equation
\begin{equation}\label{eq:degen}
\frac{p}{q}(m-\tilde{m})=(\tilde{n}-n)(n+\tilde{n})
\end{equation}
has an infinite number of solutions
$\{m,\tilde{m},n,\tilde{n}\}\in\N_{0}^{2}\times\N^{2}$ satisfying
$$
\tilde{n}-n=p,\quad m-\tilde{m}=q(n+\tilde{n}),
$$
in other words, every eigenvalue of $H[\xi]$ has infinite
multiplicity.

On the other hand, in the case $\theta\in\R_{+}\setminus\Q$ the
spectrum $H[\xi]$ is simple but it becomes `denser' as the energy
increases. In other words, the eigenvalue gaps have no positive
lower bound. Indeed, by Dirichlet's approximation theorem, for all
$N\in\N$ there exist $p,q\in\N$ such that $q\leq N$ and
$|q\theta-p|\leq(N+1)^{-1}$. We will be concerned about large values
of $N$, so we may assume that $p\geq 3$. The equation
(\ref{eq:degen}) has a infinite number of solutions
$\{m,\tilde{m},n,\tilde{n}\}\in\N_{0}^{2}\times\N^{2}$ satisfying
$$\tilde{n}+n=p,\quad 0<|\tilde{n}-n|\leq 2,\quad m-\tilde{m}=q(\tilde{n}-n).$$
For any of these solutions we obtain
\begin{align*}
|\theta m+n^2-(\theta \tilde{m}+\tilde{n}^2)| &= \left|\left(\frac{p}{q}+\theta-\frac{p}{q}\right)(m-\tilde{m})+n^2-\tilde{n}^2\right| \\
&= \left| \theta-\frac{p}{q}\right||m-\tilde{m}|
\leq \frac{1}{q(N+1)}\,2q=\frac{2}{N+1}.
\end{align*}
We infer that for any $\varepsilon>0$ there are infinitely many
pairs of eigenvalues $\lambda_{m,n},\,\lambda_{\tilde{m},\tilde{n}}$
with the property that
$$
|\lambda_{m,n}-\lambda_{\tilde{m},\tilde{n}}|<\varepsilon
$$
which proves our claim.

\subsubsection{Inclined layer parallel with the magnetic field}

The situation changes when the tilted layer has the right angle with
the original one becoming thus parallel to the field direction. Then
we have $\gamma=\frac{\pi}{2}$, and therefore
\begin{equation}\label{eq:decomp}
H[\xi]=-\partial_{s}^{2}+\left(\xi-B_0 au\right)^2-a^{-2}\partial_{u}^{2}=\overline{T_1\otimes I + I\otimes T_2[\xi]},
\end{equation}
where
$$
T_1:=-\partial_{s}^{2},\quad T_2[\xi]:=-a^{-2}\partial_{u}^{2}+\left(\xi-B_0 au\right)^2,
$$
and we decomposed $L^{2}(\R\times I,\dd s\dd u)=L^2(\R,\dd s)\otimes
L^2(I,\dd u)=:\mathscr{H}_1\otimes\mathscr{H}_2$. The other choice $\gamma=-\frac{\pi}{2}$ gives rise to a unitarily equivalent operator. Since
$\sigma(T_1)=\sigma_{ac}(T_1)=[0,+\infty)$ and $T_2[\xi]$ has a
positive simple pure point spectrum,
$\sigma(H[\xi])=[\inf\sigma(T_2[\xi]),+\infty)$ for all $\xi\in\R$.
Moreover according to Ref.~\onlinecite{Da_06} the absolute continuity of one of the operators in the decomposition \eqref{eq:decomp} implies the absolute continuity of the full operator. Since Ref.~\onlinecite{Da_06} does not contain any proof of this claim nor any references to it, we decided to include the following simple reasoning.

Let $E_1,\, E_2$, and $E$ be spectral families of the operators
$T_1,\,T_2[\xi]$ and $H[\xi]$, respectively. Then for all
decomposable vectors  $f_1\otimes
f_2\in\mathscr{H}_1\otimes\mathscr{H}_2$ we have
\begin{multline*}
 \|E(t) (f_1\otimes f_2)\|^2=\int_{-\infty}^{\infty}\int_{-\infty}^{t-s}\dd_u\|E_1(u)f_1\|^2 \dd_s\|E_2(s)f_2\|^2\\
 =\int_{\R^2}\chi_{(-\infty,t]}(u+s)\dd_u\|E_1(u)f_1\|^2 \dd_s\|E_2(s)f_2\|^2\\
 =\|E_1(\cdot)f_1\|^2 * \|E_2(\cdot)f_2\|^2((-\infty,t]);
\end{multline*}
for the first equality we refer here to the proof of Theorem
8.34 in Ref.~\onlinecite{weidmann}. Using the Fubini theorem it is easy to check that
the convolution of an absolutely continuous measure with another
(Borel) measure is also absolutely continuous. Consequently, the
absolutely continuous subspace of
$\mathscr{H}_1\otimes\mathscr{H}_2$ contains all finite linear
combinations of the decomposable vectors. However, these functions
form a dense subspace, and moreover, the absolutely continuous
subspace is always closed. This allows us to infer that $H[\xi]$ is
purely absolutely continuous.

Let $(\mu_n[\xi],\varphi_n[\xi]),\, n\in\N$, denote the eigenpairs
of $T_2[\xi]$, where the eigenvalues are numbered in the ascending
order. Some important properties of the $\mu_n[\xi]$'s are reviewed
in Ref.~\onlinecite{GeSe_97}. In particular, they are  even and
strictly increasing for $\xi>0$, thus they have the only stationary
point at $\xi=0$. Since $\sigma(H[\xi])=[\mu_1[\xi],+\infty)$,
$\sigma(H)=\sigma_{ac}(H)=[\mu_1[0],+\infty)$. Hence the bottom of
$\sigma(H)$ is given by the first eigenvalue of the harmonic
oscillator constrained to the line segment $I$. This question was
addressed repeatedly in the literature, see e.g. Ref.~\onlinecite{De_66}, and
it is easy to see that the answer is given by the smallest solution
of the equation
\begin{equation} \label{spectcon}
 {}_1 F_1\left(-\frac{\mu}{4 B_0}+\frac{1}{4},\frac{1}{2},B_0 a^2\right)=0
\end{equation}
with respect to $\mu$. Here ${}_1 F_1$ stands for the Kummer
confluent hypergeometric function. If we denote this solution
$\mu(B_0,a)$, then $\mu(B_0,a)=a^{-2} \mu(B_0 a^2,1)$, hence it is
sufficient to inspect the dependence of $\mu$ on one of the
parameters. The solution to the spectral condition \eqref{spectcon}
cannot be written in a closed form but can be found numerically, see
Fig.~\ref{fig:bottom}. Moreover, using known asymptotic properties
of the Kummer functions \cite{De_66} one can find the behavior of
$\mu(B_0,1)$,
\begin{equation}\label{eq:bottom_asy}
\mu(B_0,1)=\frac{\pi^2}{4}+\left(\frac{1}{3}-\frac{2}{\pi^2}\right)B_{0}^{2}+\left(\frac{4}{45\pi^2}-\frac{20}{3\pi^4}+\frac{56}{\pi^6}\right)B_0^4+\mathcal{O}(B_{0}^{5})
\end{equation}
as $B_0\to 0$.

\begin{figure}[hbt]
  \begin{center}
  \includegraphics{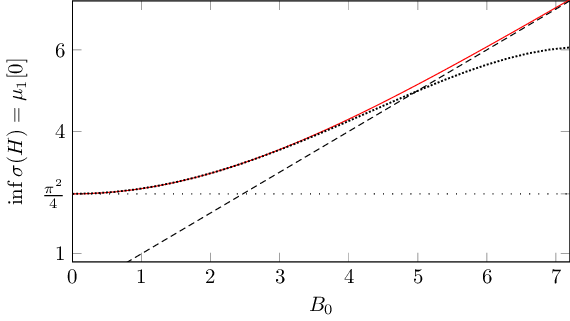}
  \end{center}
  \caption{Bottom of $\sigma(H)$ as a function of $B_0$ (red solid line); $a=1$. The black loosely dotted line corresponds to the first Dirichlet eigenvalue on $(-1,1)$, i.e., $\frac{\pi^2}{4}$. The black densely dotted line corresponds to the asymptotic expansion \eqref{eq:bottom_asy} at $B_0=0$ and the black dashed line to the asymptotic expansion \eqref{eq:bottom_strongasy} at $B_0=+\infty$.}  \label{fig:bottom}
\end{figure}

To derive the asymptotic behavior of $\mu(B_0,1)$ for large values of
$B_0$ we begin with the variational characterization of the lowest
eigenvalue,
$$
\mu(B_0,1)=\mu_1[0]|_{a=1}=\inf_{\psi\in H_0^1(I),\|\psi\|_I=1}\left(\|\psi'\|^2_I+B_0^2\|u\psi\|^2_I\right).
$$
Let $\psi_0$ be the ground state of the one-dimensional harmonic
oscillator, see \eqref{eq:HOef} for its explicit normalized form,
then clearly
\begin{equation}\label{eq:HO_low_bound}
 \mu(B_0,1)>\inf_{\psi\in H^1(\R), \|\psi\|_\R=1}\left(\|\psi'\|^2_\R+B_0^2\|u\psi\|^2_\R\right)=\|\psi_0'\|^2_\R+B_0^2\|u\psi_0\|^2_\R=B_0.
\end{equation}
If we put $\tilde\psi(u):=\psi_0(u)-\psi_0(1),\, u\in I,$ then
$\tilde\psi\geq 0,\ \tilde\psi_0\in H_0^1(I)$ and
\begin{multline*}
 \mu(B_0,1)\leq\frac{1}{\|\tilde\psi\|_I^{2}}\langle\tilde\psi,T_2[0]|_{a=1}\tilde\psi\rangle_I=\frac{1}{\|\tilde\psi\|_I^{2}}\langle\tilde\psi,B_0\tilde\psi+B_0\psi_0(1)-B_0^2\psi_0(1)u^2\rangle_I\\
 \leq \frac{1}{\|\tilde\psi\|_I^{2}}\langle\tilde\psi,B_0\tilde\psi+B_0\psi_0(1)\rangle_I=B_0+B_0\psi_0(1)\frac{1}{\|\tilde\psi\|_I^{2}}\langle\tilde\psi,1\rangle_I\\
 \leq B_0+B_0\psi_0(1)\frac{1}{\|\tilde\psi\|_I}|I|^{1/2}.
\end{multline*}
Next, for any fixed $\delta\in(0,1)$ we have
$$
\|\tilde\psi\|_I\geq
\|\psi_0\|_I-\|\psi_0(1)\|_I=\|\psi_0\|_I-\psi_0(1)|I|^{1/2}\geq
1-\delta
$$
provided $B_0$ is sufficiently large. We conclude that, for any
$\tilde\delta>0$,
\begin{equation*}
 \mu(B_0,1)\leq B_0+(1+\tilde\delta)B_0\psi_0(1)|I|^{1/2}=B_0+(1+\tilde\delta)\frac{\sqrt{2}}{\pi^{1/4}}B_{0}^{5/4}\mathrm{e}^{-B_0/2}
\end{equation*}
as $B_0\to+\infty$. Putting this together with
\eqref{eq:HO_low_bound} we arrive at the expansion
\begin{equation}\label{eq:bottom_strongasy}
\mu(B_0,1)=B_0 + \mathcal{O} (B_{0}^{5/4}\mathrm{e}^{-B_0/2})
\end{equation}
valid as $B_0\to +\infty$.

In turn, $H$ is also purely absolutely continuous by Ref.~\onlinecite[Thm
XIII.85]{RS4}. The present situation is particular because $H$ is
also invariant with respect to $s$-translations, i.e. in the
$z$ direction. Using the partial Fourier--Plancherel transform
in $s$, we find that $\tilde{H}$ is unitarily equivalent to the
operator having the following direct integral decomposition
\begin{equation*}
 \int_{\R}^{\oplus}T[\eta]\dd\eta:=\int_{\R}^{\oplus}(\eta^2+T[0])\,    \dd\eta,\text{ where }T[0]=\left(-i\partial_y-B_0 au\right)^2-a^{-2}\partial_{u}^{2}.
\end{equation*}
The physical contents of this decomposition is obvious: $T[0]$ is a
purely absolutely continuous operator describing the
edge-state-induced transport in a two-dimensional Dirichlet strip
due to the perpendicular magnetic field
\cite{BrRaSo_07,GeSe_97,HiSo_08}, and so is $T[\eta]$. The
difference between $T[\eta]$ and $T[0]$ which `adds' to the absolute
continuity of $H$ is the square of the momentum in the $z$ direction
where the motion is free. 

The decomposition \eqref{eq:di_decomp} is also reflected in the
unitary propagator for $H$ which is given by a direct integral of
the unitary propagators for $H[\xi],\,\xi\in\R$. Since each $H[\xi]$
is separated, we have
\begin{equation*}
 \exp{(-itH[\xi])}=\overline{\exp{(-itT_1)}\otimes\exp{(-itT_2[\xi])}}
\end{equation*}
by Ref.~\onlinecite[Thm 8.35]{weidmann}. In the $z$ direction the evolution
governed by the propagator with the well-known kernel
$$
(\exp{(-itT_1)})(s,s')=(4\pi i t)^{-1/2}\exp{\frac{(s-s')^2}{4i t}},
$$
while the second propagator in the tensor product decomposition can
be expressed in terms of the eigenpairs of $T_2[\xi]$ describing the
edge states as
$$
\exp{(-itT_2[\xi])}=\sum_{n\in\N}\mathrm{e}^{-it\mu_n[\xi]}\vert\varphi_n[\xi]\rangle\langle\varphi_n[\xi]\vert.
$$
The latter describes the well-known edge-current dynamics
\cite{HiSo_08}, the additional degree of freedom is a free motion of
the wavepackets in the $z$ direction having the usual properties, in
particular, the spreading with time \cite[Sec.~9.3]{BEH}.

\subsection{Bent and asymptotically flat layers} \label{sec:bent}

Let us return now to geometrically nontrivial perturbations of the
layer and assume they are localized at any fixed $y$ cut, i.e., that
the layer is flat for $s$ outside a bounded interval $(\tilde s_0,s_0)$. Specifically,
we suppose that $\dot{x}(s)=\alpha_{+}>0$ holds for all $s\geq s_0$ and $\dot{x}(s)=\alpha_{-}>0$ holds for all $s\leq\tilde s_0$. Due to \eqref{eq:param}, $\alpha_\pm\le 1$ and we may put  $\dot{z}(s)=\sqrt{1-\alpha_{+}^{2}}$ and $\dot{z}(s)=\sqrt{1-\alpha_{-}^{2}}$ for all $s\geq s_0$ and $s\leq\tilde{s}_0$, respectively. Here we have chosen the signs of $\dot z(s)$ outside $(\tilde s_0,s_0)$ just for definiteness -- with other choices the considerations below would be exactly the same. Therefore, we have
\begin{equation*}
\tilde{A_2}(s,u)=\begin{cases}B_0\left(\alpha_{+}(s-s_0)+x(s_0)-au\sqrt{1-\alpha_{+}^{2}}\right)&         		  s\geq s_0\\
                  B_0\left(\alpha_{-} (s-\tilde s_0)+x(\tilde s_0)-au\sqrt{1-\alpha_{-}^2}\right)& s\leq\tilde{s}_0.
                 \end{cases}
\end{equation*}
Note that the  positivity of $\alpha_\pm$ means that we exclude the situation where the layer is
asymptotically parallel with the magnetic field. With the future
purpose in mind we also assume that  $\alpha_{+} \ne \alpha_{-}$; without loss of generality
we may suppose that $\alpha_{-} > \alpha_{+}$, since in the
case $\alpha_{-} < \alpha_{+}$ it is sufficient to change the layer parametrization
replacing $s$ by $-s$. 

Let $s_1= s_1(\xi)$ be
the unique solution of the equation
$$
\xi+B_0\big(\alpha_{+}(s_1-s_0)+x(s_0)\big)=0.
$$
It is straightforward to
check that $\lim_{\xi\to -\infty}s_1(\xi)=+\infty$ and
$$
\xi+\tilde{A}_2=B_0\left(\alpha_{+}(s-s_1)-au\sqrt{1-\alpha_{+}^2}\right)
$$
holds for all $s\geq s_0$. Consequently, the fiber $H[\xi]$ acts for
all $s\geq s_0$ in the same way as the following positive operator,
$$
H_{\alpha_{+},s_1}(B_0)\equiv H_{\alpha_{+}, s_1}=-\partial_{s}^{2}+B_{0}^{2}\left(\alpha_{+}(s-s_1)-au\sqrt{1-\alpha_{+}^2}\right)^2-a^{-2}\partial_{u}^{2}.
$$
Next we introduce the unitary transform $U_{s_1}:\,
\psi(s,u)\mapsto \psi(s-s_1,u)$, then
$$
U_{s_1}^{-1}H_{\alpha_{+},s_1}U_{s_1}=H_{\alpha_{+},0}=:H_{\alpha_{+}}\equiv H_{\alpha_{+}}(B_0)
$$
and $H_{+}[\xi]:=U_{s_1}^{-1}H[\xi]U_{s_1}$ acts as
\begin{equation*}
  H_{+}[\xi]=
 \begin{cases}
 H_{\alpha_{+}} &\quad s\geq s_0-s_1\\[.7em]
 \begin{aligned}&-\partial_{s}f_a(s+s_1,u)^{-2}\partial_{s}+(\xi+\tilde{A}_{2}(s+s_1,u))^2\\
 &-a^{-2}\partial_{u}^{2}+V(s+s_1,u)\end{aligned} &\quad s<s_0-s_1.
 \end{cases}
\end{equation*}

\begin{lemma}\label{lem:top_eq}
Let $s_1=s_1(\xi)$ be as above. Then, for all $\xi$
sufficiently negative, there exist constants  $C_{\pm}(\xi)$ and
$K_{\pm}\in\R$, the latter being independent of $\xi$, such that
$0<C_{-}(\xi)<C_{+}(\xi)$,
\begin{equation}\label{eq:squeeze}
 C_{-}(\xi)H_{\alpha_{+}}+K_{-}\leq H_{+}[\xi]\leq C_{+}(\xi)H_{\alpha_{+}}+K_{+},
\end{equation}
and the $C_{\pm}(\xi)$ have finite positive limits as
$\xi\to-\infty$.
\end{lemma}
\begin{proof}
For all $s<s_0-s_1$ we have
 \begin{equation}\label{eq:up_bound}
   H_{+}[\xi]\leq-(1-a\|\kappa\|_\infty)^{-2}\partial_{s}^{2}+(\xi+\tilde{A}_{2}(s+s_1,u))^2-a^{-2}\partial_{u}^{2}+\|V\|_\infty.
 \end{equation}
Since $V$ is continuous and compactly supported, $\|V\|_\infty<+\infty$.
Given $s\in\R$, put $f(s):=\xi+\tilde{A}_{2}(s+s_1,0)$, then
$f(0)=0$, $f(s)=B_0\alpha_{+} s$ on $(s_0-s_1,+\infty)$,
$f(s)=f(\tilde s_0-s_1)+B_0\alpha_{-}(s-\tilde s_0+s_1)$ on
$(-\infty,\tilde s_0-s_1)$, and $\sup_{s,s'\in(\tilde
s_0-s_1,s_0-s_1)}|f(s)-f(s')|=\sup_{s,s'\in(\tilde
s_0,s_0)}|\tilde{A}_{2}(s,0)-\tilde{A}_{2}(s',0)|<+\infty$. Hence
for all $\xi$ large enough negative, i.e. for $s_1$ sufficiently
large there exist $\tilde C_\pm=\tilde C_\pm(\xi)$ such that
$0<\tilde C_{-}<\tilde C_{+}$ and $\tilde C_{+} s\leq f(s)\leq\tilde
C_{-} s$ on $(-\infty,0)$. Moreover, since $\sup_{s\in\R,u\in I}|
\xi+\tilde A_2(s+s_1,u)-f(s)|=\sup_{s\in\R,u\in I}|\tilde
A_2(s,u)-\tilde A_2(s,0)|<+\infty$, we conclude that for all $\xi$
sufficiently large negative there are $\hat C_\pm=\hat C_\pm(\xi)$
such that $0<\hat C_{-}<1<\hat C_{+}$, and
\begin{align*}
\hat C_{+}B_0\Big(\alpha_{+} s-au\sqrt{1-\alpha_{+}^2}\Big)\leq
\xi&+\tilde{A}_{2}(s+s_1,u) \\
&\leq \hat C_{-}B_0\Big(\alpha_{+}
s-au\sqrt{1-\alpha_{+}^2}\Big)<0
\end{align*}
holds on $(-\infty,s_0-s_1)\times I$. A closer inspection shows that
$\hat C_\pm$ may be chosen in such a way that
\begin{equation*}
 \hat C_{-}(\xi)\nearrow\frac{\min\{\alpha_{+},\alpha_{-}\}}{\alpha_{+}}=1, \quad \hat C_{+}(\xi)\searrow\frac{\max\{\alpha_{+},\alpha_{-}\}}{\alpha_{+}}=\frac{\alpha_{-}}{\alpha_{+}}
\end{equation*}
holds as $\xi\to -\infty$. Now we can proceed with estimate \eqref{eq:up_bound},
\begin{multline*}
  H_{+}[\xi]\leq -(1-a\|\kappa\|_\infty)^{-2}\partial_{s}^{2}+\hat C_{+}^{2}B_0^2(\alpha_{+} s-au\sqrt{1-\alpha_{+}^2})^2-a^{-2}\partial_{u}^{2}+\|V\|_\infty\\
 \leq \max\left\{(1-a\|\kappa\|_\infty)^{-2}, \hat C_{+}^{2}\right\}H_{\alpha_{+}}+\|V\|_\infty,
\end{multline*}
for all $s<s_0-s_1$. For $s\geq s_0-s_1$, this bound holds
trivially, too. In a similar manner one can estimate $H_{+}[\xi]$
from below putting
\begin{equation*}
C_{-}(\xi)=\min\left\{(1+a\|\kappa\|_\infty)^{-2}, \hat C_{-}(\xi)^{2}\right\},\quad K_{-}=-\|V\|_\infty. \qedhere
\end{equation*}
\end{proof}

If $\xi$ is sufficiently large positive the argument is a simple modification
of the above one. We define $\tilde s_1=\tilde s_1(\xi)$ as the
unique solution of the equation
$$
\xi+B_0(\alpha_{-}(\tilde s_1-\tilde s_0)+x(\tilde s_0))=0.
$$
Clearly, $\lim_{\xi\to +\infty}\tilde s_1(\xi)=-\infty$ and the
operator $H_{-}[\xi]:=U_{\tilde s_1}^{-1}H[\xi]U_{\tilde s_1}$ acts
as
\begin{equation*}
 H_{-}[\xi]=
 \begin{cases}
 H_{\alpha_{-}} &\quad s\leq \tilde s_0-\tilde s_1\\[.7em]
 \begin{aligned}&-\partial_{s}f_a(s+\tilde s_1,u)^{-2}\partial_{s}+(\xi+\tilde{A}_{2}(s+\tilde s_1,u))^2\\
 &-a^{-2}\partial_{u}^{2}+V(s+\tilde s_1,u)\end{aligned} &\quad s>\tilde s_0-\tilde s_1.
 \end{cases}
\end{equation*}
The operator pair $H_{-}[\xi]$ and $H_{\alpha_{-}}$ satisfies an
estimate analogous to \eqref{eq:squeeze}.

Now we are in position to prove an important convergence result.
\begin{proposition}\label{prop:res_conv}
Let $\mu>\|V_{-}\|_\infty$, then we have
$$
\lim_{\xi\to \mp\infty}\big\|(
H_{\pm}[\xi]+\mu)^{-1}-(H_{\alpha_{\pm}}+\mu)^{-1}\big\|=0.
$$
\end{proposition}
\begin{proof}
It is sufficient to apply Theorem 2.3 of Ref.~\onlinecite{Tu_16}. To make the paper
self-contained we reproduce this result here: let $\{A[\alpha]|\,
\alpha\in(-\infty,+\infty]\}$ be a one parametric family of
lower--bounded self-adjoint operators on $L^2(M)$, where
$M\subset\R^n$ is open, with the following properties
 \begin{enumerate}[\upshape (i)]
  \item \label{prop:core} $C_{0}^{\infty}(M)$ is a core of $A[\alpha]$ for all $\alpha\in(-\infty,+\infty]$.
  \item \label{prop:low_bound} There exist $C>0$ and $K,\,\alpha_0\in\R$ such that, for all $\alpha\geq\alpha_0$, $C A[+\infty]+K\leq A[\alpha]$.
  \item \label{prop:conv} For any compact set $\mathcal{K}\subset M$, there exists $\alpha_{\mathcal{K}}$ such that, for all $\alpha\geq\alpha_{\mathcal{K}}$, $A[\alpha]|_{C_{0}^{\infty}(\mathcal{K})}=A[+\infty]|_{C_{0}^{\infty}(\mathcal{K})}$.
  \item \label{prop:comp} $A[+\infty]$ has compact resolvent.
 \end{enumerate}
Then, for any $z\in\mathrm{Res}(A[+\infty])$ and $\varepsilon>0$,
there exists $\alpha_{z,\varepsilon}$ such that for all
$\alpha>\alpha_{z,\varepsilon}$, $z\in\mathrm{Res}(A[\alpha])$ and
 $$
 \big\|(A[\alpha]-z)^{-1}-(A[+\infty]-z)^{-1}\big\|<\varepsilon.
 $$
To deal with the limit $\xi\to-\infty$ we put
$\alpha=-\xi,\,M=\R\times I,\,A[\alpha]=H_{+}[-\xi],$ and
$A[+\infty]=H_{\alpha_{+}}$. The properties (i) and (ii) above are
direct consequences of Lemma~\ref{lem:top_eq}, (iii) is obvious from
the definition of $H_{+}[\xi]$, and (iv) was proved in
Sec.~\ref{sec:comp}. The limit $\xi\to+\infty$ is treated in a
similar manner.
\end{proof}

Let us denote the eigenvalues of $H_{\alpha_{+}}$, arranged in the
ascending order with the multiplicity taken into account, by
$\sigma_n(\alpha_{+})\equiv\sigma_n(\alpha_{+},B_0)$. Since the
norm-resolvent convergence implies the convergence of eigenvalues,
we see that in any neighborhood of $\sigma_n(\alpha_{+})$, there is
exactly the same number of eigenvalues of $H[\xi]$ as is the
multiplicity of $\sigma_n(\alpha_{+})$ in the spectrum of
$H_{\alpha_{+}}$, provided $\xi$ is chosen sufficiently large
negative. Similarly, in any neighborhood of $\sigma_n(\alpha_{-})$,
there is exactly the same number of eigenvalues of $H[\xi]$ as is
the multiplicity of $\sigma_n(\alpha_{-})$ in the spectrum of
$H_{\alpha_{-}}$, provided $\xi$ is positive and sufficiently large.
Moreover, if we fix $E>0$ then for all $\sigma_n(\alpha_{+})$ less
than $E$ we may choose the said neighborhoods to be disjoint and to
prove that in the remaining gaps there are no eigenvalues of
$H[\xi]$ for all $\xi$ sufficiently large negative. Again, a similar
statement holds true for large positive values of $\xi$.

In general, it may occur that the eigenvalue branches of $H[\xi]$
cross. It cannot happen, however, that a non-constant eigenvalue
branch crosses a constant branch. In fact, if there is a constant
eigenvalue branch then it has to be isolated from the rest of the
spectrum, cf. Remark~\ref{rem:spec_prop} above. Hence it makes sense
to denote it as $\lambda_m[\xi]$, since it is indeed the $m$th
eigenvalue of $H[\xi]$, with the multiplicity taken into account,
for some $m\in\N$ and all $\xi\in\R$. If it is independent of $\xi$
we would have $\lambda_m[\xi]=\sigma_m(\alpha_{+},B_0)
=\sigma_m(\alpha_{-},B_0)$; our aim now is to find a sufficient
condition under which this cannot happen.

To this aim, let us denote
$$
T_{\alpha}(s,u):=B_0^2\left(\alpha s-a\sqrt{1-\alpha^2}\, u\right)^2,
$$
and fix an $\varepsilon>0$. Then, since
$$
|a s u|\leq\frac{1}{2}\left(\varepsilon s^2+\varepsilon^{-1}a^2u^2\right),
$$
we get the inequalities
\begin{multline*}
\left(\alpha^2-\alpha\sqrt{1-\alpha^2}\,\varepsilon\right)s^2-a^2\varepsilon^{-1}\alpha\sqrt{1-\alpha^2}\leq
B_{0}^{-2}T_{\alpha}(s,u)\\ \leq
\left(\alpha^2+\alpha\sqrt{1-\alpha^2}\,\varepsilon\right)s^2+a^2\left(1-\alpha^2+\varepsilon^{-1}\alpha\sqrt{1-\alpha^2}\right).
\end{multline*}
This allows us to infer that
\begin{align*}
 &H_{\alpha_{+}}(B_0)\leq H_{1}(B_0\,g(\alpha_{+},\varepsilon))+B_0^2 a^2\left(1-\alpha_{+}^2+\varepsilon^{-1}\alpha_{+}\sqrt{1-\alpha_{+}^2}\right)\\
 &H_{\alpha_{-}}(B_0)\geq H_1(B_0\,g(\alpha_{-},-\varepsilon))-B_0^2 a^2\varepsilon^{-1}\alpha_{-}\sqrt{1-\alpha_{-}^2},
\end{align*}
where $g(\alpha,\varepsilon):= \sqrt{\alpha^2
+\alpha\sqrt{1-\alpha^2}\,\varepsilon}$.

With respect to the first terms on the right-hand sides of the above
inequalities, recall that $\sigma_m(1,B)=\lambda_{\tilde m,\tilde
n}(B)=B(2\tilde m+1)+a^{-2}E_{\tilde n}$ for some $\tilde m\in\N_0$
and $\tilde n\in\N$, where
\begin{equation} \label{eq:D_ev}
E_{\tilde n}:=\left(\frac{\tilde n\pi}{2}\right)^2.
\end{equation}
Here, $E_{\tilde n}$ is the $\tilde n$th `transverse' Dirichlet eigenvalue for
$a=1$. Furthermore, note the monotonicity with respect to the field:
if $B<\tilde B$, then
\begin{equation}\label{eq:ev_diff}
\sigma_m(1,\tilde B)-\sigma_m(1,B)\geq \tilde B-B
\end{equation}
holds for all $m\in\N$. This follows from the fact that we have
\begin{eqnarray*}
\lefteqn{|\{\lambda_{\tilde m,n}(\tilde B)|\, \exists \tilde
m\in\N_0:\, \lambda_{\tilde m,n}(\tilde B)\leq E\}|} \\[.3em] && \leq
|\{\lambda_{\tilde m,n}(B)|\, \exists \tilde m\in\N_0:\,
\lambda_{\tilde m,n}(B)\leq E-(\tilde B-B)\}|,
\end{eqnarray*}
for all $E>0$ and $n\in\N$, where $|\cdot|$ stands for the
cardinality of a set. Now using the minimax principle we obtain
\begin{align*}
 &\sigma_m(\alpha_{+}, B_0)\leq\sigma_m(1,B_0\,g(\alpha_{+},\varepsilon))+B_0^2 a^2\left(1-\alpha_{+}^2+\varepsilon^{-1}\alpha_{+}\sqrt{1-\alpha_{+}^2}\right)\\
 &\sigma_m(1,B_0\,g(\alpha_{-},-\varepsilon))\leq\sigma_m(\alpha_{-},B_0)+B_0^2 a^2\varepsilon^{-1}\alpha_{-}\sqrt{1-\alpha_{-}^2}.
\end{align*}
Combining this with \eqref{eq:ev_diff} we arrive at the following
claim.
\begin{lemma}\label{lem:ev_diff}
Let $\alpha_{-} >\alpha_{+}>0$ and
 $$
 \varepsilon_0:=\frac{\alpha_{-}^2-\alpha_{+}^2}{\alpha_{+}\sqrt{1-\alpha_{+}^2}+\alpha_{-}\sqrt{1-\alpha_{-}^2}}.
 $$
Then we have $g(\alpha_{-},-\varepsilon)>g(\alpha_{+},\varepsilon)$
for all positive $\varepsilon<\varepsilon_0$. If, in addition,
 \begin{equation*}
  a<a_0(\varepsilon):=\frac{1}{\sqrt{B_0}}\,\sqrt{\frac{g(\alpha_{-},-\varepsilon)-g(\alpha_{+},\varepsilon)}{1-\alpha_{+}^2+\varepsilon^{-1}\left(\alpha_{+}\sqrt{1-\alpha_{+}^2}+\alpha_{-}\sqrt{1-\alpha_{-}^2}\right)}},
 \end{equation*}
holds for some $\varepsilon<\varepsilon_0$, then
$\sigma_m(\alpha_{+}, B_0)<\sigma_m(\alpha_{-}, B_0)$.
\end{lemma}
\begin{remark} \label{rem:a_crit_bound}
One cannot maximize the threshold $a_0$ with respect to
$\varepsilon\in(0,\varepsilon_0)$ analytically. However, it is
possible to find a closed-form estimate by maximizing a lower bound.
First of all, note that $0\leq\alpha_{\pm}\sqrt{1-\alpha_{\pm}^2}\leq
1/2$ holds for all $\alpha_{\pm}\in(0,1]$. Hence
$\varepsilon_0\geq\tilde\varepsilon_0:=\alpha_{-}^2-\alpha_{+}^2$,
where $\tilde\varepsilon_0<1$, and
 \begin{align*}
 \sqrt{B_0}\,a_0(\varepsilon) &\geq\sqrt{\frac{\sqrt{\alpha_{-}^2-\varepsilon/2}-\sqrt{\alpha_{+}^2+\varepsilon/2}}{2\varepsilon^{-1}}} \\
 &=\frac{1}{\sqrt{2}}\,\sqrt{\frac{\varepsilon(\alpha_{-}^2-\alpha_{+}^2-\varepsilon)}{\sqrt{\alpha_{-}^2-\varepsilon/2}+\sqrt{\alpha_{+}^2+\varepsilon/2}}}
 \geq\frac{1}{\sqrt{2}}\sqrt{\frac{\varepsilon(\alpha_{-}^2-\alpha_{+}^2-\varepsilon)}{\alpha_{-}+\sqrt{\alpha_{+}^2+1/2}}},
 \end{align*}
for all $\varepsilon<\tilde\varepsilon_0$. The bound is maximal for
$\varepsilon=(\alpha_{-}^2-\alpha_{+}^2)/2 <\tilde\varepsilon_0$, so
we arrive at
 $$
 \sup_{\varepsilon\in(0,\varepsilon_0)}a_0(\varepsilon)
 \geq \frac{\alpha_{-}^2-\alpha_{+}^2}{2\sqrt{2B_0}\sqrt{\alpha_{-}+\sqrt{\alpha_{+}^2+1/2}}}
 \geq \frac{\alpha_{-}^2-\alpha_{+}^2}{2\sqrt{2}\,\sqrt{1+\sqrt{3/2}}\,\sqrt{B_0}}.$$
\end{remark}

\medskip

By a \emph{reductio ad absurdum} we can thus make the following
conclusion.
\begin{proposition}
 Under the assumptions of Lemma \ref{lem:ev_diff}, there are no constant eigenvalue branches of $H[\xi]$. Therefore, the spectrum of $-\Delta_{D,A}^\Omega$ is purely absolutely continuous.
\end{proposition}

\subsection{Thin layers} \label{sec:thin}

The result of Sec.~\ref{sec:bent} involved already a restriction on
the layer thickness possibly going beyond the assumption \eqref{eq:diff_cond},
the severity of which depended on how much the layer was `broken'.
Now we will go further and look what sufficient condition can
be derived if the layer is very thin.

To begin with, recall that it was proved in Ref.~\onlinecite{KrRaTu_15} that if, in addition to \eqref{eq:curv_bound} and \eqref{eq:diff_cond}, 
 \begin{equation} \label{eq:conv_ass} 
 \dot\kappa,\, \ddot\kappa\in L^\infty 
 \end{equation}
then for any $k$ large enough,
$$
\big\|\big(\tilde{H}-a^{-2}E_1+k\big)^{-1}-(\tilde
h_{\mathrm{eff}}+k)^{-1}\oplus 0\big\|=\mathcal{O}(a)
$$
holds as $a\to 0+$, with
\begin{align*}
 \tilde h_{\mathrm{eff}}=-\partial_{s}^{2}+(-i\partial_{y}+B_{0}x(s))^2-\frac{1}{4}\kappa^{2}(s)
\end{align*}
acting on $L^2(\R^2,\dd s\dd y)$. Recall that $E_1$ is given by \eqref{eq:D_ev}. Also remark that \eqref{eq:curv_bound} combined with \eqref{eq:conv_ass} yields $V\in L^\infty$, i.e., the second part of \eqref{eq:V_bound}. Since we assume that the curvature
$\kappa$ is bounded, $\tilde{h}_{\mathrm{eff}}$ is essentially
self-adjoint on $C_{0}^{\infty}(\R^2)$ and self-adjoint on
$$
\text{Dom}~\tilde{h}_{\mathrm{eff}}=\big\{ f\in L^2(\R^2,\dd s\dd
y) \mid -\partial_{s}^{2}f+(-i\partial_{y}+B_{0}x(s))^2 f\in
L^2(\R^2,\dd s\dd y) \big\},
$$
see Ref.~\onlinecite{LeSi_81}.

\begin{remark}[Magnetic field]
The vector potential in the operator $\tilde h_{\mathrm{eff}}$ is
$\tilde{A}_{\mathrm{eff}}=(0,B_{0}x(s))$, and consequently,
$$
\tilde{B}_{\mathrm{eff}}=\mathrm{curl}\tilde
{A}_{\mathrm{eff}}=B_{0}\dot x(s)=B\cdot n.
$$

\end{remark}

\medskip

\noindent  Using the partial Fourier--Plancherel transform in the $y$ variable,
we turn the operator $\tilde h_{\mathrm{eff}}$ into
\begin{equation*}
h_\mathrm{eff}:=-\partial_{s}^{2}+(\xi+B_{0}x(s))^2-\frac{1}{4}\kappa^{2}(s).
\end{equation*}
It is self-adjoint on its
definition domain,
$$
\text{Dom}~h_{\mathrm{eff}}=\big\{ f\in L^2(\R^2,\dd s\dd \xi) \mid
-\partial_{s}^{2}f+(\xi+B_{0}x(s))^2 f\in L^2(\R^2,\dd s\dd \xi)\big\},
$$
and decomposes into a direct integral,
$$
h_{\mathrm{eff}}=\int^\oplus_\R h_{\mathrm{eff}}[\xi]\dd \xi,
$$
where the fiber is given by
$$
h_{\mathrm{eff}}[\xi]=-\partial_{s}^{2}+(\xi+B_{0}x(s))^2-\frac{1}{4}\kappa^{2}(s)
$$
as an operator on $L^2(\R,\dd s)$. Using Ref.~\onlinecite[Thm 6.3]{KrRaTu_15},
we obtain further
$$
\big\|(\tilde H-a^{-2}E_1 +k)^{-1}-(\tilde H_0-a^{-2}E_1 +k)^{-1}\big\|= \mathcal{O}(a)
$$
as $a\to 0+$, where $\tilde H_0$ is the `leading term',
$$
\tilde H_0:=\tilde{h}_{\mathrm{eff}}\otimes I+I\otimes (-a^{-2}\partial_{u}^{2}),
$$
and $k$ has to be, of course, chosen large enough. In view of the
unitarity of the Fourier--Plancherel transform, this implies
\begin{equation*}
 \|(H-a^{-2}E_1 +k)^{-1}-(H_0-a^{-2}E_1 +k)^{-1}\|= \mathcal{O}(a),
\end{equation*}
where $H_0$ is defined similarly as $\tilde H_0$ but now with the help of
$h_{\mathrm{eff}}$. Since this operator also decomposes into a
direct integral,
$$
H_0=\int^\oplus_\R \left(h_{\mathrm{eff}}[\xi]\otimes I+I\otimes (-a^{-2}\partial_{u}^{2})\right)\,\dd\xi=:\int^\oplus_\R  H_0[\xi]\dd\xi,
$$
we obtain the corresponding limiting relation for the fibers,
\begin{equation}\label{eq:thin_conv}
 \|(H[\xi]-a^{-2}E_1 +k)^{-1}-(H_0[\xi]-a^{-2}E_1 +k)^{-1}\|= \mathcal{O}(a)
\end{equation}
as $a\to 0+$. This follows from the fact that $\|\int^{\oplus}_{M}
A[\xi]\|=\esssup_{M}\|A[\xi]\|$, cf.~Ref.~\onlinecite[Thm XIII.83]{RS4}, which
also implies, in particular, that the error term on the right-hand
side of \eqref{eq:thin_conv} is uniform in $\xi\in\R$.

Assume that the operator $h_{\mathrm{eff}}[\xi]$ has compact
resolvent and all its eigenvalues are simple and analytic in $\xi$.
This is fulfilled if, for instance, in addition to
\eqref{eq:curv_bound}, which is sufficient for analyticity, we have
$\underline{\dot x}_\pm>0$, which beside compactness of the resolvent assures simplicity of the spectrum\cite{Tu_16}. Here we employ for the
sake of brevity the notation
\begin{align*}
 \underline{f}_{+}&:=\sup_{a\in\R}\essinf_{t\in(a,+\infty)}f(t) & \overline{f}_{+}&:=\inf_{a\in\R}\esssup_{t\in(a,+\infty)}f(t)\\
 \underline{f}_{-}&:=\sup_{a\in\R}\essinf_{t\in(-\infty,a)}f(t) & \overline{f}_{-}&:=\inf_{a\in\R}\esssup_{t\in(-\infty,a)}f(t).
\end{align*}
for a given $f\in L^\infty(\R;\R)$. We denote the eigenvalues of
$h_{\mathrm{eff}}[\xi]$, arranged in the ascending order, as
$\nu_m[\xi],\, m\in\N$. Assume that they are non-constant as
functions of $\xi$. 

Under the stated assumptions on $h_{\mathrm{eff}}[\xi]$, the spectrum of $H_0
[\xi]-a^{-2}E_1$ consists of isolated eigenvalues
$$
\gamma_{m,n}[\xi]=\nu_m[\xi]+a^{-2}(E_n-E_1).
$$
By the minimax principle, $\nu_m[\xi]\geq
-\frac14\,\|\kappa^2\|_\infty$. We fix an energy value $E\in\R$
which we will refer to for brevity as threshold, then there exists
an $a_E>0$ such that
$$
E< a^{-2}(E_2-E_1)-\frac{1}{4}\|\kappa^2\|_\infty
$$
holds for all $a< a_E$. Consequently, for these values of $a$, the spectrum of
$H_0[\xi]-a^{-2}E_1$ strictly below $E$ consists of
the simple eigenvalues $\gamma_{m,1}[\xi]=\nu_m[\xi],\,
m=1,2,\ldots, N[\xi]$, only. Note that $\max_{\xi\in\R}N[\xi]=:
N_E<+\infty$.

There exist at least one compact interval $I_E$ and a $\delta_E>0$
such that $\nu_{N_E}[\xi]<E-3\delta_E$ holds for all $\xi\in I_E$,
because $\nu_{N_E}[\xi]$ is by assumption non-constant and analytic.
For a fixed $m=1,\ldots, N_E$ we can then construct a tubular
neighborhood $\mathcal{T}_m(\delta):=\{\nu_m[\xi]+t|\: \xi\in I_m,\,
t\in (-\delta,\delta)\}$ with $\delta<\min\{\delta_E,
\tilde{\delta}_E\}$, where
$$\tilde{\delta}_E:= \frac14\, \min_{m=1,\ldots,N_E}\inf_{\xi\in I_E}\mathrm{dist}\left(\nu_m[\xi],
\sigma(H_0[\xi]-a^{-2}E_1)\right)
$$
is strictly positive. Furthermore, one can find an $\tilde{a}_E\in
(0,a_E)$ such that for all $a < \tilde{a}_E$ there is exactly one
eigenvalue branch of $H-a^{-2}E_1$ passing through each of the neighborhoods
$\mathcal{T}_m(\delta)$, $m=1,\ldots, N_E$. Since $\nu_m[\xi]$ are
non-constant, these eigenvalue branches must be non-constant as well, if we choose $\delta$ and consequently
also $\tilde{a}_E$ small enough.

Assume that there is a constant eigenvalue branch of
$H-a^{-2}E_1$ below $E-\delta_E$. Then it must be, in
particular, constant in the interval $I_E$, and thus it could not
intersect with any of $\mathcal{T}_m(\delta)$, provided we chose
$\delta$ and $a$ as above. Moreover, by an easy perturbation theory
consideration there are no eigenvalues of $H[\xi]-a^{-2}E_1$ in the
remaining gaps whenever $a$ is small enough. From this we can
conclude that for any fixed threshold $E$, all the eigenvalue branches
of $H$ that lie (at least partially) below
$E+a^{-2}E_1$ are non-constant provided the layer halfwidth $a$ is
sufficiently small.

In the argument above, it was crucial that all the $\nu_m$'s were non-constant as functions of $\xi$. 
A sufficient condition for this may be found in Ref.~\onlinecite{Tu_16},
\begin{equation}\label{eq:eff_suff_cond_1}
\underline{\dot{x}}_{\pm}>0\,\wedge\,\underline{\dot{x}}_{+}\geq\overline{\dot{x}}_{-}\,\wedge\,\left(\underline{\kappa^2}_{+}-\overline{\kappa^2}_{-}<4
B_0(\underline{\dot{x}}_{+}-\overline{\dot{x}}_{-})\right).
\end{equation}
(Alternatively, one can change the $\pm$ indices to $\mp$ everywhere in \eqref{eq:eff_suff_cond_1}.) One can think, e.g.,  about  a layer that is asymptotically flat with different magnitudes of the asymptotic slopes at $x=\pm\infty$. 

Now, we are going to derive another condition for non-constancy of $\nu_m$'s.
Let us assume that, in addition to \eqref{eq:curv_bound},
\begin{equation}\label{eq:eff_suff_cond_2}
x(s)=s \quad\text{for } s\leq 0 ;\quad\dot{x}(s)\geq 0
\quad\text{for } s>0;\quad  \underline{\dot{x}}_{+}>0;\quad x\neq
\mathrm{Id}.
\end{equation}
It is convenient to write $x(s)=s+r(s)$, where $r(s)=\theta(s)r(s)$
with $\theta$ being the Heaviside step function. Clearly, $r(s)\leq 0$ and it is not identically
zero. Note that this
includes any perturbation of the planar layer that is compact in the
$x$ direction, since without loss of generality we may always
suppose that such a perturbation is supported to the right of the
origin. (See Fig.~\ref{fig:thin} for examples.)
\begin{figure}[ht] 
\centering
\includegraphics[scale=0.7]{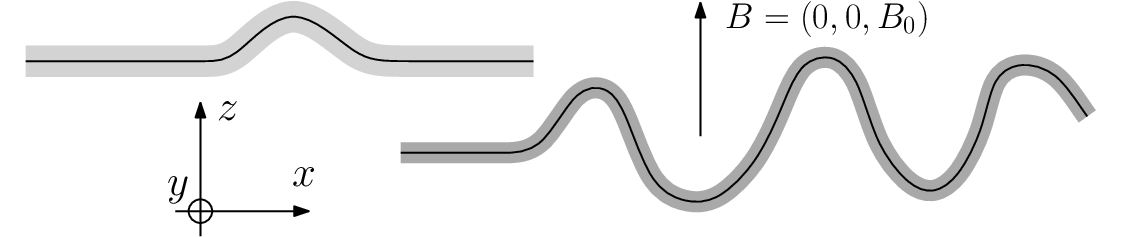} 
\caption{Examples of layers obeying \eqref{eq:eff_suff_cond_2}.} \label{fig:thin}
\end{figure}
We have
\begin{multline} \label{eq:HO_pert_decomp}
h_{\mathrm{eff}}[\xi]=-\partial_{s}^{2}+(\xi+B_0 s)^2+2B_0(\xi+B_0 s)r(s)+B_0^2r^2(s)-\frac{1}{4}\kappa^2(s)\\
=h_{\text{HO}}[\xi]+\theta(s)\left(B_0 r(s)\left(2\xi+B_0(s+x(s))\right)-\frac{1}{4}\kappa^2(s)\right).
\end{multline}
For any $\xi\geq 0$, $h_{\mathrm{eff}}[\xi]$ is thus a non-positive
perturbation of the shifted harmonic oscillator Hamiltonian
$h_{\text{HO}}[\xi]$. Let us estimate the eigenvalue $\nu_m[0]$.

Let $\psi_j$ be the $j$th eigenfunction of the harmonic oscillator
Hamiltonian $h_{\text{HO}}[0]$ given by \eqref{eq:HOef} and
$S_m:=\text{span}\{\psi_j|\: j=1,\ldots,m\}$; note that every
function in $S_m$ is real analytic on $\R$. Now, by the minimax
principle,
\begin{multline} \label{eq:minimax_est}
 \nu_m[0]=\min_{\substack{S\subset \dom{h_{\mathrm{eff}}[0]}\\\dim S=m}}\ \max_{\substack{\psi\in S \\ \|\psi\|=1}}\langle\psi,h[0]\psi\rangle\leq \max_{\substack{\psi\in S_m \\ \|\psi\|=1}}\langle\psi,h[0]\psi\rangle\\
 \leq \max_{\substack{\psi\in S_m \\ \|\psi\|=1}}\langle\psi,h_{\text{HO}}[0]\psi\rangle+\max_{\substack{\psi\in S_m \\ \|\psi\|=1}}\langle\psi,\theta(s)(B_0^2 r(s)(s+x(s))-\frac14\kappa^2(s))\psi\rangle\\
 =B_0(2m+1)+\max_{\substack{\psi\in S_m \\ \|\psi\|=1}}\int_{0}^{+\infty}\big(B_0^2 r(s)(s+x(s))-\frac14\kappa^2(s)\big)|\psi(s)|^2\dd s.
\end{multline}
The last term on the right-hand side is negative, because the
sub-integral function is non-positive everywhere and strictly
negative on some interval, and the maximum of the integral is
attained for some $\psi_{\text{max}}\in S_m$. Indeed, if the maximum
was zero then $\psi_{\text{max}}$ would be zero on the mentioned
interval, and therefore due to the analyticity it would vanish on
$\R$, which is a contradiction. We conclude that the sharp
inequality $\nu_m[0]<B_0(2m+1)$ holds for all $m\in\N$.

On the other hand, we have $\lim_{\xi\to +\infty}\nu_m[\xi]
=B_0(2m+1)$. To prove this claim we start with the unitary transform
$U_\xi:\psi(s)\mapsto\psi(s-\xi/B_0)$ and introduce
\begin{equation*}
 \hat{h}_{\mathrm{eff}}[\xi]=U_\xi h_{\mathrm{eff}}[\xi] U_{\xi}^{-1}=\begin{cases}
                                         -\partial_{s}^{2}+B_{0}^{2}s^2 & s\leq
                                         \xi/B_0\\[.7em]
                                         \begin{aligned}
                                         &-\partial_{s}^{2}+(\xi+B_0 x(s-\xi/B_0))^2 \\ &-\frac{1}{4}\kappa^2(s-\xi/B_0)\end{aligned} & s>\xi/B_0
                                        \end{cases}
\end{equation*}
and put $\hat{h}_{\mathrm{eff}}[+\infty]:=h_{\text{HO}}[0]$. Now we
may apply the result of Ref.~\onlinecite[Thm 2.3]{Tu_16}, which we have
reproduced here as a part of the proof of
Proposition~\ref{prop:res_conv}, to the family
$\{\hat{h}_{\mathrm{eff}}[\xi]|\: \xi\in(-\infty,+\infty]\}$. Let us
focus on the assumption (ii) of the theorem. For all $s$
sufficiently large we have $\dot{x}(s)>\frac12
\underline{\dot{x}}_{+}$, and consequently, there is a $\xi_0>0$
such that for all $\xi>\xi_0$ and $s>0$,
$$
\xi+B_0 x(s-\xi/B_0)>\frac12\, B_0 \underline{\dot{x}}_{+} s.
$$
Using this estimate on the interval $(\xi/B_0,+\infty)$, we obtain
\begin{equation*}
 \frac{\underline{\dot{x}}_{+}^{2}}{4}\, \hat{h}_{\mathrm{HO}}[0]-\frac{1}{4}\|\kappa\|_{\infty}^{2}\leq \hat{h}_{\mathrm{eff}}[\xi].
\end{equation*}
The remaining assumptions are easy to verify. This makes it possible
to infer that
$$
\lim_{\xi\to+\infty}\|\hat{h}_{\mathrm{eff}}[\xi]^{-1}-h_{\mathrm{HO}}[0]^{-1}\|=0,
$$
which in turn implies that $\lim_{\xi\to +\infty}\nu_m[\xi]$ is just
the $m$th eigenvalue of $h_{\text{HO}}[0]$. Our findings are
summarized in the following claim.
\begin{proposition}\label{prop:thin}
Let the assumptions \eqref{eq:curv_bound}, \eqref{eq:diff_cond}, and \eqref{eq:conv_ass} hold together with either
\eqref{eq:eff_suff_cond_1} or \eqref{eq:eff_suff_cond_2}. Then to
any $E\in\R$ one can find an $a_E>0$ such that no eigenvalue branch of
the total Hamiltonian $H$ that lies at least partially below
$E+a^{-2}E_1$ can be constant as a function of $\xi$ whenever
$a<a_E$.
\end{proposition}

\begin{corollary}
 Under the assumptions of Proposition \ref{prop:thin}, for any $E\in\R$ there exists $a_E>0$ such that if $a<a_E$ then the spectral measure of $-\Delta_{D,A}^\Omega$ on $(-\infty,E+a^{-2}E_1)$ is purely absolutely continuous, i.e., the spectrum of $-\Delta_{D,A}^\Omega$ below $E+a^{-2}E_1$ is absolutely continuous. (Note that without shifting $E$ by $a^{-2}E_1$ the result would be void because $\inf\sigma(-\Delta_{D,A}^\Omega)=a^{-2}E_1+\mathcal{O}(1)$ as $a\to 0+$.)
\end{corollary}

\begin{remark}\label{rem:eff}
 Let us stress that, in Proposition \ref{prop:thin}, instead of \eqref{eq:eff_suff_cond_1} or \eqref{eq:eff_suff_cond_2} one could impose any other condition that would imply that at least those eigenvalues of $h_\mathrm{eff}[\xi]$ which lie below $E$ are not constant as functions of $\xi$.
\end{remark}

\section{ An extension of the Iwatsuka model} \label{s: iwatsuka}

While our main interest concerns magnetic transport in the Dirichlet
layers, the considerations at the end of Sec.~\ref{sec:thin}, in particular,
the decomposition of the type \eqref{eq:HO_pert_decomp} can be in
combination with the minimax principle applied also to the classic
Iwatsuka model. We start with the two-dimensional Hamiltonian
 \begin{equation} \label{eq:Iwatsuka}
  h_{\mathrm{Iw}}=-\partial^2_x+(-i\partial_y+A_y(x))^2+W(x),
 \end{equation}
 where
 $$
 A_y(x)=\int_0^x B(t)\dd t.
 $$
Fix a $B_0>0$ and assume that $B(t)=B_0(1+b(t))$ with
 \begin{enumerate}[(i)]
  \item \label{ass:1}$b\in L^2_{\text{loc}}(\R)$,
  \item $b(t)=0$ for all $t<0$,
  \item \label{ass:3} $\int_0^x b(t)\dd t\leq 0$ holds for all $x\geq 0$,
  \item \label{ass:4} there are $\alpha\in(-1,0)$, $x_1\geq 0$ such that $\int_0^x
b(t)\dd t>\alpha x$ holds for all $x\geq x_1$.
 \end{enumerate}
The potential $W\in L^\infty(\R) \cap C(\R)$ is such that
$W(x)=\theta(x)W(x)\leq 0$. Under the stated integrability
assumptions on $b$ and $W$, the operator $h_{\mathrm{Iw}}$ is essentially
self-adjoint on $C_0^\infty (\R^2)$\cite{LeSi_81}.

\begin{theorem} \label{theo:iwatsuka}
 Adopt the above assumptions together with $b\not\equiv 0 \vee W\not\equiv 0$, then $h_{\mathrm{Iw}}$ is purely absolutely continuous.
\end{theorem}
 \begin{proof}
As in the seminal paper of Iwatsuka \cite{Iw_85} we start with a
direct integral decomposition into fiber operators
  \begin{equation*}
   h_{\mathrm{Iw}}[\xi]=-\partial^2_x+(\xi+A_y(x))^2+W(x)
  \end{equation*}
on $L^2(\R,\dd x)$. In the same manner as in Ref.~\onlinecite{Tu_16} we show
that $h_{\mathrm{Iw}}[\xi]$ has compact resolvent and all its eigenvalues numbered
in the ascending order as $\lambda_m[\xi],\, m\in\N,$ are simple and
analytic on $\R$ as functions of $\xi$. To prove the absolute
continuity of $h_{\mathrm{Iw}}$ it suffices to demonstrate that no
$\lambda_m[\cdot]$ is constant. We have
  \begin{equation} \label{eq:Iw_ext_decomp}
   h_{\mathrm{Iw}}[\xi]=h_{\text{HO}}[\xi]+\theta(x)\Big(B_0 r(x)\big(2\xi+B_0(2x+r(x))\big)+W(x)\Big),
  \end{equation}
  where $h_{\text{HO}}[\xi]$ was introduced in $\eqref{eq:HO_pert_decomp}$ and
  \begin{equation*}
   r(x):=\int_0^x b(t)\dd t.
  \end{equation*}
Note that $r(x)$ is in view of \eqref{ass:1} (absolutely) continuous. Using
\eqref{ass:4}, we find
$$
2\xi+B_0(2x+r(x))> 2\xi+B_0\big((2+\alpha)x+R\big)
$$
for all $x\geq 0$, where $R:=\min_{x\in[0,x_1]} r(x)\leq 0$. Hence,
if we set $\xi=-\frac12\,B_0 R$ then
$2\xi+B_0(2x+r(x))=B_0(2x+r(x)-R)\geq 0$ holds for all $x\geq 0$.
Taking the non-positivity of $r$ and $W$ into the account, we
conclude that the second term in \eqref{eq:Iw_ext_decomp} is also
non-positive. Moreover, the assumption $b\not\equiv 0 \vee
W\not\equiv 0$ implies that it is strictly negative on some
interval.

Mimicking the estimates in \eqref{eq:minimax_est} with $S_m$ being
the span of first $m$ eigenfunctions of $h_\text{HO}[-B_0R/2]$, we arrive at
$\lambda_m[-B_0R/2]<B_0(2m+1)$. Due to \eqref{ass:4}, the second
assumption of Theorem 2.3 in Ref.~\onlinecite{Tu_16} is fulfilled which finally
implies that $\lim_{\xi\to+\infty}\lambda_m[\xi]=B_0(2m+1)$.
 \end{proof}

Let us recall that the family of magnetic fields considered above
has a non-empty intersection with all the families studied earlier
in Ref.~\onlinecite{Iw_85,MaPu_97,ExKo_00, Tu_16} which, with the
exception of the last one, treat the original Iwatsuka model,
$W\equiv 0$. Hence we obtain a nontrivial extension of the known
results, with notably weak regularity assumptions comparing to the
other sources. Note also the assumption \eqref{ass:3} crucial for
the use of the minimax principle does not mean that the perturbation
$b$ of the constant magnetic field must be everywhere negative; it
may be sign-changing and negative on a compact set only.

\subsection*{Acknowledgments}

The research has been partially supported by the Czech Science
Foundation (GA\v{C}R) within the project No. 17-01706S.



\begin{thebibliography}{99}

\bibitem{BEH}
J.~Blank, P.~Exner, M.~Havl\'{\i}\v{c}ek: \emph{Hilbert Space Operators in Quantum Physics}, 2nd extended edition, Springer, Dordrecht 2008.
\bibitem{BrRaSo_07}
P. Briet, G. Raikov, E. Soccorsi: Spectral properties of a magnetic quantum Hamiltonian on a strip, \emph{Asymp. Anal.} \textbf{58} (2008), 127--155.
\bibitem{CFKS}
H.L.~Cycon, R.G.~Froese, W.~Kirsch, B.~Simon: \emph{ Schr\"odinger Operators, with Applications to Quantum Mechanics and Global
Geometry}, Springer, Berlin and Heidelberg 1987.
\bibitem{Da_06}
M. Damak: On the spectral theory of tensor product Hamiltonians, \emph{J. Oper. Theory} \textbf{55} (2006), 253--268.
\bibitem{De_66}
P. Dean: The constrained quantum mechanical harmonic oscillator, \emph{Proc. Camb. Phil. Soc.} \textbf{62} (1966), 277--286.
\bibitem{DoGeRa_11} N. Dombrowski, F. Germinet, G. Raikov: Quantization of edge currents along magnetic barriers and magnetic guides, \emph{Ann. H. Poincar\'{e}} {\bf 12} (2011), 1169--1197.
\bibitem{DoHiSo_14}
N. Dombrowski, P.D. Hislop, E. Soccorsi: Edge currents and eigenvalue estimates for magnetic barrier Schr\"{o}dinger operators, \emph{Asymp. Anal.} {\bf 89} (2014), 331--363.
\bibitem{DuExKr_01}
P. Duclos, P. Exner, D. Krej\v{c}i\v{r}\'{i}k: Bound states in curved quantum layers, \emph{Commun. Math. Phys.} \textbf{223} (2001), 13--28.
\bibitem{ExKo_00} P. Exner, H. Kova\v{r}\'{i}k: Magnetic strip waveguides, \emph{J. Phys. A: Math. Gen.} {\bf 33} (2000), 3297--3311.
\bibitem{EK}
P.~Exner, H.~Kova\v{r}\'{\i}k: \emph{Quantum Waveguides}, Springer, Cham 2015.
\bibitem{FiSo_06}
N. Filonov, A.V. Sobolev: Absence of the singular continuous component in spectra of analytic direct integrals, \emph{J. Math. Sci.} \textbf{136} (2006), 3826--3831.
\bibitem{GeSe_97}
V. Geiler, M. Senatorov: Structure of the spectrum of the
Schr\"{o}dinger operator with magnetic field in a strip and
infinite-gap potentials, \emph{Sb. Math.} \textbf{188} (1997),
657--669.
\bibitem{HiPoRaSu_16}
P.D. Hislop, N. Popoff, N. Raymond, M.P. Sundqvist: Band functions in the presence of magnetic steps, \emph{Math. Models \& Methods in Applied Science} {\bf 26} (2016), 161--184.
\bibitem{HiSo_08a}
P. Hislop, E. Soccorsi: Edge Currents for Quantum Hall Systems, I.
One-Edge, Unbounded Geometries, \emph{Rev. Math. Phys.} \textbf{20}
(2008), 71--115.
\bibitem{HiSo_08}
P. Hislop, E. Soccorsi: Edge Currents for Quantum Hall Systems, II.
Two-Edge, Bounded and Unbounded Geometries, \emph{Ann.
H.~Poincar\'{e}} \textbf{9} (2008), 1141--1175.
\bibitem{HiSo_15}
P.D. Hislop, E. Soccorsi: Edge states induced by Iwatsuka Hamiltonians with positive magnetic fields, \emph{J. Math. Anal. Appl.} {\bf 422} (2015),  594--624.
\bibitem{Iw_85}
A. Iwatsuka: Examples of absolutely continuous Schr\"{o}dinger operators in magnetic fields, \emph{Publ. RIMS, Kyoto Univ.} \textbf{21} (1985), 385--401.
\bibitem{kato}
T. Kato: \emph{Perturbation Theory for Linear Operators}, 3rd Ed., Springer, 1995.
\bibitem{KrRaTu_15}
D. Krej\v{c}i\v{r}\'{i}k, N. Raymond, M. Tu\v{s}ek: The magnetic Laplacian in shrinking tubular neighbourhoods of hypersurfaces,  \emph{J. Geom. Anal.} \textbf{25} (2015), 2546--2564.
\bibitem{LeSi_81}
H. Leinfelder, C.G. Simader: Schr\"{o}dinger operators with singular magnetic potentials, \emph{Math. Zeitschrift} \textbf{176} (1981), 1--19.
\bibitem{MaPu_97} M.~M\v{a}ntoiu, R.~Purice: Some propagation properties of the Iwatsuka model, \emph{Commun. Math. Phys.} {\bf188} (1997), 691--708.
\bibitem{RS4}
M. Reed, B. Simon: \emph{Methods of Modern Mathematical Physics IV}, Academic Press, New York 1978.
\bibitem{Tu_16}
M. Tu\v{s}ek: On an extension of the Iwatsuka model, \emph{J. Phys. A: Math. Theor.} \textbf{49} (2016), 365205.
\bibitem{weidmann}
J. Weidmann: \emph{Linear Operators in Hilbert Spaces},
Springer-Verlag, New York 1980.

\end{thebibliography}
\end{document}